\documentclass[a4paper,11pt]{article}

\usepackage{caption}
\usepackage{subcaption}
\usepackage[l2tabu, orthodox]{nag}
\usepackage[utf8]{inputenc}
\usepackage[pdfusetitle]{hyperref}
\usepackage[affil-it]{authblk}

\usepackage{mathtools}

\mathtoolsset{
showonlyrefs=true 
}

\usepackage{amsfonts,amssymb}
\usepackage{amsmath,amssymb,amsthm}

\newtheorem{proposition}{Proposition}
\newtheorem{lemma}{Lemma}
\newtheorem{theorem}{Theorem}
\newtheorem{claim}{Claim}
\newtheorem{corollary}{Corollary}

\usepackage{fullpage} 

\usepackage{enumerate}

\usepackage{xspace}

\newcommand{\problem}{\textsc{Biased Graph Cleaning}\xspace}
\newcommand{\localproblem}{\textsc{Rooted Biased Graph Cleaning}\xspace}

\usepackage{tikz}
\usetikzlibrary{positioning,decorations,fit,shapes}
\makeatletter
\tikzset{circle split part fill/.style  args={#1,#2}{%
 alias=tmp@name, 
  postaction={%
    insert path={
     \pgfextra{%
     \pgfpointdiff{\pgfpointanchor{\pgf@node@name}{center}}%
                  {\pgfpointanchor{\pgf@node@name}{east}}%
     \pgfmathsetmacro\insiderad{\pgf@x}
      \fill[#1] (\pgf@node@name.base) ([xshift=-\pgflinewidth]\pgf@node@name.east) arc
                          (0:180:\insiderad-\pgflinewidth)--cycle;
      \fill[#2] (\pgf@node@name.base) ([xshift=\pgflinewidth]\pgf@node@name.west)  arc
                           (180:360:\insiderad-\pgflinewidth)--cycle;            
         }}}}}
 \makeatother

\newcommand{\N}{\ensuremath{\mathbb N}\xspace}

\newcommand{\Q}{\ensuremath{\mathbb Q}\xspace}

\newcommand{\cB}{\ensuremath{\mathcal B}\xspace}
\newcommand{\cC}{\ensuremath{\mathcal C}\xspace}
\newcommand{\cP}{\ensuremath{\mathcal P}\xspace}

\newcommand{\maybeqed}{}
\date{}

\newcommand{\opt}{\mathsf{opt}}

\begin{document}

\title{LP-branching algorithms based on biased graphs\footnote{A
    preliminary version was presented at SODA 2017. This version
    contains additional material on approximation algorithms.}}

\author[1]{Euiwoong Lee\thanks{\texttt{euiwoong@cims.nyu.edu}}}
\author[2]{Magnus Wahlstr\"om\thanks{\texttt{Magnus.Wahlstrom@rhul.ac.uk}}}

\affil[1]{Computer Science Department, New York University, USA}
\affil[2]{Dept of Computer Science, Royal Holloway, University of London, UK}

\maketitle

\begin{abstract}
We give a combinatorial condition for the existence of efficient, LP-based FPT algorithms 
for a broad class of graph-theoretical optimisation problems. 
Our condition is based on the notion of biased graphs known from matroid theory. 
Specifically, we show that given a biased graph $\Psi=(G,\cB)$, where $\cB$ is a class
of balanced cycles in $G$, the problem of finding a set $X$ of at most $k$
vertices in $G$ which intersects every unbalanced cycle in $G$ 
admits an FPT algorithm using an LP-branching approach, similar to 
those previously seen for VCSP problems (Wahlstr\"om, SODA 2014).
Our algorithm has two parts. First we define a \emph{local problem},
where we are additionally given a root vertex $v_0 \in V$ and asked
only to delete vertices $X$ (excluding $v_0$) so that the connected component of $v_0$ in $G-X$
contains no unbalanced cycle. We show that this local problem admits
a persistent, half-integral LP-relaxation
with a polynomial-time solvable separation oracle,
and can therefore be solved in FPT time via LP-branching,
assuming only oracle membership queries for the class of balanced cycles in $G$.
We then show that solutions to this local problem
can be used to tile the graph, producing an optimal solution to the 
original, global problem as well.
This framework captures many of the problems previously solved
via the VCSP approach to LP-branching, as well as new generalisations, 
such as Group Feedback Vertex Set for infinite groups
(e.g., for graphs whose edges are labelled by matrices). 
A major advantage compared to previous work is that 
it is immediate to check the applicability of the result for a given problem,
whereas
testing applicability of the VCSP approach
for a specific VCSP,
requires determining the existence of 
an embedding language
with certain algebraically defined properties,
which is not known to be decidable in general.
Additionally, we study the approximation question, and show that every
problem of this category admits an $O(\log \text{OPT})$-approximation
\end{abstract}

\section{Introduction}

In recent years, we have seen a growing interest in the use of linear or
integer linear programming methods (LP/ILP) in parameterized complexity~\cite{IwataWY16,
JansenK15ESA,GanianO18ILP,KolmanKT16SWAT}.
The appeal is clear. On the one hand, linear programming and general continuous relaxations comes with a very powerful toolbox for theoretical investigation. This promises to be a powerful hammer, especially for optimisation problems;
e.g., an FPT-size extended formulation, or more generally an FPT
separation oracle for an integral polytope would provide a way towards FPT
algorithms~\cite{KolmanKT16SWAT}.
The same can be said for Lenstra's algorithm and other methods for
solving complex problems in few variables~\cite{Lokshtanov15ILParXiv,CramptonGKW19TCS}.
On the other hand, it has been observed that ILP solvers, like SAT solvers,
frequently perform much better in practice than can be currently be
explained theoretically. It is appealing to study parameterizations
of ILP problems, to find structural reasons for this apparent 
tractability~\cite{JansenK15ESA,GanianO18ILP}.

Narrowing our focus, for many optimisation problems we have seen powerful
FPT results based on running an ILP-solver with performance guarantees.
These algorithms use LP-relaxations with strong structural properties
(\emph{persistence} and \emph{half-integrality}), which allows you to run
an FPT branching algorithm on top of the relaxation, typically getting a running time 
of $O^*(2^{O(k)})$ for a solution size of~$k$. More powerfully, these algorithms
also allow to bound the running time in terms of the \emph{relaxation gap},
i.e., the additive difference between the value of the LP-optimum
and the true integral optimum. For problems such as 
\textsc{Vertex Cover}~\cite{LokshtanovNRRS14ILP} and
\textsc{Multiway Cut}~\cite{CyganPPW13MWC}, 
this gives us FPT algorithms that find an optimum in time $O^*(2^{O(k-\lambda)})$,
where~$k$ is the size of the optimum and~$\lambda$ is the lower bound
given by the LP-relaxation of the input.
In this sense, these results can also be taken as a possible
``parameterized explanation'' of the success of ILP solvers for these
problems, although this approach of course only works for certain problems
as it is in general NP-hard to decide whether an ILP problem
has a solution matching the value of the LP-relaxation.

The main restriction for this approach is to find LP-relaxations with the
required structural properties (or even to find problems which admit such
LP-relaxations). In the cases listed above, this falls back on classical
results from approximation~\cite{NemhauserT75,GargVY04}, but further
examples proved elusive. A significant advancement was made 
relatively recently~\cite{IwataWY16} by recasting the search for LP-relaxations with
the required properties in algebraic terms, using a connection between
such LP-relaxations and so-called \emph{Valued CSPs} (VCSPs; see below). 
Using this connection, FPT LP-branching algorithms were provided
that significantly improved the running times for a range of problems,
including \textsc{Subset Feedback Vertex Set} and
\textsc{Group Feedback Vertex Set} in $O^*(4^k)$ time for solution size~$k$,
improving on previous records of~$O^*(2^{O(k \log k)})$ time~\cite{CyganPPW13,CyganPP16GFVS},
and \textsc{Unique Label Cover} in $O^*(|\Sigma|^{2k})$ time 
for alphabet $\Sigma$ and solution size~$k$, improving on 
a previous result of~$O^*(|\Sigma|^{O(k^2 \log k)})$~\cite{ChitnisCHPP16SICOMP}.
\textsc{Group Feedback Vertex Set} in particular is a meta-problem that 
includes many independently studied problems as special cases.

However, powerful though these results may be, the nature of the framework
makes it difficult to apply to a given combinatorial problem.
To do so would involve two steps. First, the problem must be phrased
as a VCSP. A VCSP asks to minimise the value of an
objective function $f(\phi)$ over assignments $\phi: V \to D$
from a finite domain $D$, where the objective $f(\phi)$ is in turn
usually given as a sum of bounded-size \emph{cost functions} $\sum_{i=1}^m f_i(\phi)$. 
To find such a formulation is not always easy, and sometimes it may be
impossible to capture a problem precisely. As an example, it is possible
to phrase \textsc{Feedback Vertex Set} as an instance of 
\textsc{Group Feedback Vertex Set} and hence arguably as a VCSP~\cite{IwataWY16},
but to discover such a formulation from first principles is not easy.
Second, given a desired VCSP formulation, it must be determined
whether the VCSP admits a \emph{discrete relaxation}, 
so that the framework can be applied~\cite{IwataWY16}.
It is not known in general how to decide the existence of such a relaxation.\footnote{The
results of~\cite{IwataWY16} are obtained by working backwards from known 
relaxations with the required properties, a list which includes 
separable $k$-submodular relaxations and arbitrary bisubmodular relaxations, 
and with a small extension can be made to cover a class of problems also
including so-called \emph{skew bisubmodular} functions~\cite{HuberKP14SICOMP}.
However, again, even given a specific target class it is usually at least
intuitively non-obvious whether a VCSP can be relaxed into the class or not.}

In short, although the results are powerful, they are also somewhat
inscrutable.\footnote{Since the publication of the preliminary version
  of this paper~\cite{Wahlstrom14SODA}, there have been a few further
  developments in this direction.  We note two in particular.
  Iwata, Yamaguchi and Yoshida~\cite{IwataYY18FOCS} provide a
  combinatorial algorithm that replaces the LP-solving step in the
  above applications, implying linear-time FPT algorithms for all of
  the problems cited above.  It is however currently not known whether
  such an algorithm exists for the general \problem problem considered
  in this paper.  In another direction, Reidl and the second author~\cite{ReidlW18ICALP}
  considered a relaxed version of the conditions for a discrete
  relaxation, thereby allowing the FPT LP-branching framework to be
  applied to further problems.}

In this paper, we instead give a combinatorial condition under which a
graph-theoretical problem admits an LP-relaxation with the required properties,
and hence an efficient FPT algorithm parameterized either by solution size,
or in particular cases by a relaxation gap parameter.
Our condition is based on the class of so-called \emph{biased graphs},
which are combinatorial objects of importance especially to matroid theory~\cite{Zaslavsky89Bias1,Zaslavsky91Bias2}.
We review these next.

\paragraph{Biased graphs.} 
A \emph{biased graph} is a pair~$\Psi=(G=(V,E), \cB)$ of a graph~$G$
and a set~$\cB \subseteq 2^E$ of simple cycles in~$G$, referred to
as the \emph{balanced cycles} of~$G$. with the property that if
two cycles~$C, C' \in \cB$ form a \emph{theta graph}
(i.e., a collection of three internally vertex-disjoint paths with shared endpoints), 
then the third cycle of $C \cup C'$ is also contained in~$\cB$. 
A cycle class $\cB$ with this property is referred to as \emph{linear}.
Dually, and more important to the present paper, a simple cycle $C$
is \emph{unbalanced} if $C \notin \cB$. The definition is equivalent to saying
that if $C$ is an unbalanced cycle, and if $P$ is
a path with endpoints in $C$ which is internally edge- and vertex-disjoint from $C$,
then at least one of the two new cycles formed by $C \cup P$ is also unbalanced.
We refer to a collection $\cC$ of cycles as a \emph{co-linear} cycle class
if the complement of $\cC$ is a linear class.
We say that an induced subgraph $G[S]$ of $G$ is balanced if $G[S]$ 
contains no unbalanced cycles. 

The basic problem considered in this paper is now defined as follows: 
Given a biased graph $\Psi=(G=(V,E), \cB)$ and an integer~$k$,
find a set of vertices~$X \subseteq V$ with~$|X|\leq k$ such that~$G-X$ is balanced. 
We refer to this as the \problem problem.
Our main result in this paper is that \problem is FPT by $k$,
with a running time of $O^*(4^k)$, assuming only access to
a membership oracle for the class $\cB$
(i.e., for a cycle $C$, given as a set of edges,
we can determine whether $C \in \cB$ with an oracle call).

An important example of biased graphs are \emph{group-labelled graphs}. 
Let $G=(V,E)$ be an oriented graph, and let the edges of $G$ be labelled 
by elements from a group $\Gamma=(D,\cdot)$, such that if an edge $uv \in E$
has label $\gamma \in D$, then the edge $vu$ (i.e., $uv$ traversed in the
opposite direction) has label $\gamma^{-1}$. Then the balanced cycles
of $G$ are the cycles $C$ such that the product of the edge labels of the cycle,
read in the direction of their traversal, is equal to the identity element $1_\Gamma$
of $\Gamma$. Note that the orientation of the edges serves only to make the
group-labelling well-defined, and has no bearing on which cycles we consider. 
It is easy to verify that this defines a linear class of cycles,
and hence gives rise to a biased graph $\Psi$.

The problem \textsc{Group Feedback Vertex Set} corresponds exactly to
\problem when $\Psi$ is defined by a group-labelled graph.
However, not all biased graphs can be defined via group labels, and
moreover, some group-labelled graphs can only be defined via an infinite
group $\Gamma$~\cite{DeVosFP14bias}. More examples of biased graphs follow below. 

Biased graphs were originally defined in the context of matroid theory.
Although this connection is not important to the present paper, we
nevertheless give a brief review. 
Each biased graph $\Psi$ gives rise to two matroids, the \emph{frame matroid}
and the \emph{lift matroid} of $\Psi$. These are important examples in
structural matroid theory. The \emph{Dowling geometries} $Q_n(\Gamma)$ 
for a group $\Gamma$, originally defined by Dowling~\cite{Dowling73}, 
are equivalent to frame matroids of complete $\Gamma$-labelled multigraphs. 
For more on matroids for biased graphs, see Zaslavsky~\cite{Zaslavsky89Bias1,Zaslavsky91Bias2}, as well as 
the series of blog posts on the \emph{Matroid Union} weblog~\cite{MUBias1,MUBias2,MUBias3}.

\paragraph{Our approach.}
Inspired by the previous algorithm for \textsc{Group Feedback Vertex Set}~\cite{IwataWY16}, 
our approach for \problem consists of two parts,
the \emph{local problem} and the \emph{global problem}. 
In the local problem, the input is a biased graph $\Psi=(G=(V,E), \cB)$
together with a \emph{root vertex} $v_0 \in V$ and an integer $k$, 
and the task is restricted to finding a set $X \subseteq V$ of vertices, 
$|X|\leq k$ and $v_0 \notin X$, such that the connected component of $v_0$ 
in $G-X$ is balanced. Equivalently, the local problem can be defined as 
finding a set $S \subseteq V$ with $v_0 \in S$ such that $G[S]$ is balanced
and connected, and $|N_G(S)| \leq k$. 
We refer to this local problem as \localproblem.

We show that \localproblem can be solved via an LP which is half-integral and has a 
stability property similar to persistence. This LP uses a formulation where the
constraints correspond to rooted cycles we refer to as \emph{balloons}. 
The formulation of this slightly unusual LP is critical to the tractability
of the problem. The possibly more natural approach of letting the obstacles
of the LP simply be unbalanced cycles would not work as well; for instance,
although \textsc{Feedback Vertex Set} is an instance of \problem (with balanced
cycles $\cB=\emptyset$), it is known that the natural cycle-hitting LP has an 
integrality gap of a factor of $\Theta(\log n)$~\cite{ChudakGHW98}.\footnote{Although formulations of LPs for \textsc{Feedback Vertex Set} exist with an integrality gap of~$2$~\cite{ChudakGHW98}, we are not aware of any similarities between these formulations and our local LP.
Moreover, other special cases of \problem, e.g., \textsc{Odd Cycle Transversal},
admit no constant-factor approximation unless the Unique Games Conjecture fails.}

The properties of the LP further imply that \localproblem has a
2-approximation, even for weighted instances, and can be solved
(in the unweighted case) in time~$O^*(4^{k-\lambda})$ where~$\lambda$
is the value of the LP-optimum, assuming access to a membership oracle
for the class $\cB$ as above. 
In particular, \localproblem can be solved in time~$O^*(2^k)$. 
We note that several independently studied problems arise
as special cases of \localproblem; see below.

In order to solve the global problem, \problem, we show that the local LP
obeys a strong persistence-like property, analogous to the 
\emph{important separator} property frequently used in graph separation 
problems~\cite{Marx06MWC}, which allows us to identify ``furthest-reaching''
local connected components when solving the local problem, such that
the connected components produced by the algorithm for the local problem
can be used to ``tile'' the original graph in a solution to the global problem.
A $O^*(4^k)$-time algorithm for \problem follows
(although the ``above lower bound'' 
perspective does not carry over to solutions for the global problem).

\paragraph{Results and applications.}

We summarise the above statements in the following theorems.
Let $\Psi=(G=(V,E), \cB)$ be a biased graph, where $\cB$ is 
defined via a membership oracle
that takes as input a simple cycle $C$, provided
as an edge set, and tests whether $C \in \cB$. 
Then the following apply.

\begin{theorem} \label{thm:local}
  Assuming a polynomial-time membership oracle for the class of balanced cycles,
  \localproblem admits the following algorithmic results:
  \begin{itemize}
  \item A polynomial-time 2-approximation;
  \item An FPT algorithm with a running time of~$O^*(4^{k-\lambda})$,
    where $\lambda \geq k/2$ is the value of the LP-relaxation of the problem;
  \item An FPT algorithm with a running time of~$O^*(2^k)$.
  \end{itemize}
  The 2-approximation holds even for weighted graphs.
\end{theorem}

\begin{theorem} \label{thm:global}
  Assuming a polynomial-time membership oracle for the class of balanced cycles,
  \problem admits an FPT algorithm with a running time of $O^*(4^k)$.
\end{theorem}

To illustrate the flexibility of the notion, let us consider some classes of biased graphs.
\begin{itemize}
\item If $\cB=\emptyset$, then \problem corresponds simply to \textsc{Feedback Vertex Set}
\item If $\Psi$ arises as a $\Gamma$-labelled graph, then \problem corresponds to
  \textsc{Group Feedback Vertex Set}. If $\Gamma$ is finite, then the result is 
  equivalent to the previous LP-based algorithm~\cite{IwataWY16}.
\item If $\Psi$ is $\Gamma$-labelled for an infinite group $\Gamma$, e.g.,
  a matrix group, then previous results do not apply, since they assume the existence
  of an underlying VCSP presentation (in the correctness proofs, if not in the algorithms).
  However, the problem is still FPT, assuming, essentially, that the word problem
  for $\Gamma$ can be solved. 
\item To show a case that does not obviously correspond to group-labelled graphs,
  let $G$ be (improperly) edge-coloured, and let a cycle be balanced if and only if
  it is monochromatic. It is not difficult to see that this defines a biased graph.%
  \footnote{Yoichi Iwata (personal communication) noted that this does
    correspond to a group-labelled graph, using a quotient group of
    $Z_2^E$.}

\item Finally, Zaslavsky~\cite{Zaslavsky89Bias1} notes another case that in general 
  does not admit a group-labelled representation. Let $G$ be a copy of $C_n$,
  with two parallel copies of every edge. Let $\cB$ be a class of ``isolated''
  Hamiltonian cycles of $G$, in the sense that for any $C \in \cB$,
  switching the edge used between $u$ and $v$ for any pair of
  consecutive vertices of $G$ results in an unbalanced cycle. 
  It is not hard to verify that this defines a biased graph. 
  To keep $G$ as a simple graph, we simply subdivide the edges; this does not
  affect the collection of cycles (although it means that the term
  ``Hamiltonian'' fails to apply).
  
  However, the corresponding problem can be very difficult, e.g., 
  $\cB$ may consist of only exactly one of the $2^n$ candidate cycles,
  giving an oracle lower bound
  of $\Omega(2^k)$ against any algorithm for \problem.%
  \footnote{To be precise, it gives a lower bound of $2^k$ against the
    edge-deletion version of the problem, but it is not difficult
    to transfer such a lower bound to a lower bound against \problem.}
\end{itemize}
Finally, we show as promised that the algorithm for \localproblem
has independent applications. 
\begin{itemize}
\item Let $G=(V,E)$ be any graph, and add an apex vertex $v_0$ to $G$.
  Let $\cB=\emptyset$. Then \localproblem corresponds to \textsc{Vertex Cover},
  hence Theorem~\ref{thm:local} is an $O^*(4^k)$-time FPT algorithm for
  the problem \textsc{Vertex Cover Above Matching}, which 
  encompasses
  the more commonly known problem \textsc{Almost 2-SAT}~\cite{MishraRSSS11,RamanRS11}.
\item Let $G=(V,E)$ be a graph and $T \subseteq V$ a set of terminals. 
  Duplicate each terminal $t \in T$ into $d(t)$ copies, forming a set $T'$,
  and add a vertex $v_0$ to $G$ with $N(v_0)=T'$. Let a cycle be \emph{unbalanced}
  if and only if it passes through $v_0$ and two vertices $t, t' \in T'$ 
  which are copies of distinct terminals in $T$. Then \localproblem corresponds to 
  \textsc{Multiway Cut}, and Theorem~\ref{thm:local} reproduces the known
  2-approximation and $O^*(2^k)$-time algorithm~\cite{GargVY04,CyganPPW13MWC}.

  We note that the linear class condition prevents us from representing 
  any other cut problem this way; if terminal-terminal paths from $t$ to $t'$
  and from $t'$ to $t''$ are allowed in $G$, for some $t, t', t'' \in T$,
  then also the path from $t$ to $t''$ must be allowed. Hence the set of 
  allowed paths induces an equivalence relation on $T$. 
\end{itemize}

In general, we find that the notion of biased graphs is surprisingly subtle, 
and corresponds surprisingly well to the class of (natural) problems for which 
LP-branching FPT algorithms are known. We also note, without further study
at the moment, that 
there is significant similarity between the (local or global) problems that
can be expressed via biased graphs this way and the class of graph separation
problems for which the existence of \emph{polynomial kernels} is either
known, or most notoriously open~\cite{KratschW12FOCS}.

Finally, although our results do include some cases which were not
previously known to be FPT, 
we feel that the most significant advantage of the present work
is the transparency and naturalness of the definition. The existence of
a purely combinatorial condition also arguably brings us closer to the 
existence of a purely combinatorial algorithm for these problems. 
Since the results in this paper rely on using an
LP-solver with a separation oracle, 
a purely combinatorial algorithm could significantly decrease the hidden
polynomial factor in the running times.

\paragraph{Additional results: Approximation.}
As an additional result in the current paper (not present in the preliminary version~\cite{Wahlstrom14SODA}),
we consider the approximation properties of \problem. 
Using properties of the LP developed in the rest of the paper,
and using standard {\em region growing} arguments used for \textsc{Multicut} and \textsc{Sparsest Cut}~\cite{LeightonR99, GargVY96},
we show the following result. 

\begin{theorem} \label{thm:approx}
  Assuming a polynomial-time membership oracle for the class of balanced cycles,
  the weighted version of \problem admits a polynomial-time $O(\log k)$-approximation,
  where $k$ is the cost of an optimal solution. 
\end{theorem}

This problem does not appear to have been previously studied in the
approximation literature, even in the natural special case of
\textsc{Group Feedback Vertex Set}. 

We note that although ratios of $O(\log n)$ and $O(\log k)$ appear
similar from an approximation perspective, the latter are strongly
preferable for applications in parameterized complexity.
In particular, $f(k)$-type approximation ratios have been used
repeatedly in polynomial kernelization to bootstrap the kernelization
procedure, including for \textsc{Odd Cycle Transversal}~\cite{KratschW14TALG},
and for \textsc{Almost 2-SAT} and \textsc{Group Feedback Vertex Set}~\cite{KratschW12FOCS}.
In all of these cases, the size of the kernelization depends on the
approximation ratio when expressed as a function of $k$, and improved
ratios $f(k)$ would immediately improve the corresponding kernel
sizes.  Using Theorem~\ref{thm:approx} would improve the kernel size
computed in~\cite{KratschW12FOCS} for \textsc{Group Feedback Vertex Set},
from $O(k^{2s+2})$ vertices to $O((k \log k)^{s+1})$ vertices, where
$s$ is the size of the group. We do not pursue this connection further
in the present paper, because in preparation of a journal version
of~\cite{KratschW12FOCS} (in preparation, 2019), 
Kratsch and Wahlstr\"om already noted an $O(s \log (sk) \log \log
(sk))$-approximation by reducing to \textsc{Symmetric Multicut}
and using results by Even et al.~\cite{EvenNRS00JACM}.
Thus, for $s=O(1)$, the difference in approximation ratio is minor.
However, for $s=\omega(1)$, Theorem~\ref{thm:approx} provides a
significant improvement in the available approximation guarantee.

For more on kernelization, see the book of Fomin,
Lokshtanov, Saurabh and Zehavi~\cite{FominLSZkernelsbook}. 
We will not consider this topic any further in the present paper.

\paragraph{Preliminaries.}
We assume familiarity with the basic notions of graph theory,
parameterized complexity, and the basics of combinatorial optimisation.
For a reference on parameterized complexity, see Cygan et al.~\cite{CyganFKLMPPS15PCbook};
for all necessary material on linear programming and combinatorial
optimisation, see Schrijver~\cite{SchrijverBook}.
Other notions will be introduced as they are used.

\section{Biased graphs and the local LP}

In this section, we define the local LP, used to solve the local problem,
and give some results about the structure of min-weight obstacles in it. 
We will also show that we can optimise over the LP in polynomial time
by providing a separation oracle.  In subsequent sections, we will
derive the properties of half-integrality and persistence, and
show Theorems~\ref{thm:local} and~\ref{thm:global}.

We first introduce some additional terminology. 
Recall the definition of biased graphs from Section~1. 
Let $\Psi=(G=(V,E), \cB)$ be a biased graph. For a
simple cycle $C$, a \emph{chord path} for $C$ is a simple path with end
vertices in $C$ and internal vertices and edges disjoint from $C$. If $C$
is an unbalanced cycle in $\Psi$, a \emph{reconfiguration of $C$ by $P$}
refers to an unbalanced cycle $C'$ formed from $C$ and $P$, with $C'$
containing $P$, as is guaranteed by the definition of biased graphs. Note
that a chord path can consist of a single edge and no internal vertices;
however, it is also possible that a chord path is non-induced, e.g., for
structural purposes we may reconfigure a cycle $C$ using a chord path $P$
that contains internal vertices, even if there is a direct edge in $G$
connecting the end points of $P$. A \emph{membership oracle} for $\cB$
is a (black-box) algorithm which, for every set of edges forming a simple
cycle $C$ in $G$, will respond whether $C \in \cB$. 

For simplicity, we assume that $G$ is a simple graph, by subdividing
edges if $G$ contains parallel edges. Let us observe that this can be
done without loss of generality.

\begin{proposition} \label{prop:simplegraph}
  Let $\Psi=(G, \cB)$ be a biased graph and let $G'$ be the result of subdividing some edges in $G$. 
  Then there is a bijection between simple cycles of $G$ and simple cycles of $G'$. Furthermore,
  we can define a biased graph $\Psi'=(G', \cB')$ where $\cB'$ contains edge sets of simple cycles in $G'$ 
  which correspond to balanced cycles in $G$; and given a membership oracle for $\cB$ we 
  can define a membership oracle for $\cB'$.
\end{proposition}

As a warm-up, and to illustrate the kind of arguments we will be using, we
give a simple lemma that shows why co-linearity is a useful structural
property from the perspective of the local problem.

\begin{lemma} \label{lemma:cycleintersect}
  Let $\Psi=(G=(V,E),\cB)$ be a biased graph, and let $V_R \subseteq V$
  be a set of vertices such that $G[V_R]$ is balanced and connected.
  Let $C$ be an unbalanced cycle. Then either $C$ intersects $N(V_R)$ in at most one vertex, 
  or there exists an unbalanced cycle $C'$ such that $C'$ intersects
  $V_R$ in a non-empty simple path, $C'$ intersects $N(V_R)$ in at most two vertices, 
  and $(V(C') \setminus V_R) \subseteq (V(C) \setminus V_R)$.
\end{lemma}
\begin{proof}
  Assume that $C$ intersects $N(V_R)$ at least twice, as otherwise there is nothing to show.
  Assume first that $C$ is disjoint from $V_R$, and let $u$ and $v$ be two
  distinct members of $V(C) \cap N(V_R)$. Let $P$ be a $uv$-path with internal
  vertices in $V_R$. Reconfiguring $C$ using $P$ as a chord path results in an unbalanced cycle $C'$ 
  which intersects $V_R$, and whose intersection with $V \setminus V_R$ is 
  a subset of that of $C$, hence we may assume that the cycle $C$ intersects $V_R$.

  Now consider $V(C) \cap V_R$. Since $V_R$ is balanced, this is a collection
  of paths (rather than the entire cycle). Let $P_C$ be one such path, and
  let $u$ and $w$ be the vertices in $N(V_R)$ that it terminates at. Now, if possible,
  let $P$ be a shortest path in $V_R$ connecting $P_C$ to a vertex of
  either $(N(V_R) \cap V(C)) \setminus \{u,w\}$ or $(V(C) \setminus V(P_C)) \cap V_R$.
  In both cases, using $P$ as a chord path results in a new
  unbalanced cycle $C'$ whose intersection with $N(V_R)$ is strictly
  decreased.

  The only remaining case is that $V(C) \cap V_R$ forms a single path, and 
  $V(C) \cap N(V_R)$ consists of exactly two vertices (necessarily the
  attachment points of the path), and we are done. 
\maybeqed
\end{proof}

Properties and arguments similar to this will be used extensively in the
arguments concerning the behaviour of the LP.

\subsection{The LP relaxation}

We now define the LP-relaxation used for the \localproblem. 
The LP uses constraints we refer to as \emph{balloons};
we will see that balloons can equivalently be thought of as
pairs of paths rooted in $v_0$, or as a cycle connected to $v_0$ by a path.

Let $\Psi=(G=(V,E), \cB)$ be a biased graph and $v_0 \in V$ a distinguished vertex. 
Let $\cC$ be the corresponding class of unbalanced simple cycles. 
The \localproblem problem asks for a set $S \subseteq V$ 
with $v_0 \in S$ such that $G[S]$ is balanced and connected, and $|N(S)|$ is minimum. 
We consider the following LP-relaxation for it. The variables are $\{x_v: v \in V\}$,
with $0 \leq x_v \leq 1$ and $x_{v_0}=0$. For $C \in \cC$, a \emph{$v_0$-$C$-path} 
is a simple path $P=v_0 \ldots v_\ell$ where $v_\ell \in C$ and $v_i \notin C$ for 
$1 \leq i < \ell$. If $v_0 \in C$, then $P$ consists of the single vertex $v_0$ and no edges.
We define the weights of $C$ and $P$ as
$w(C)=\sum_{v \in C} x_v$ and $w(P)=\sum_{v \in P-v_\ell} x_v + \frac 1 2 x_{v_\ell}$,
i.e., $w(P)$ assigns coefficient $\frac 1 2$ to the endpoint of $P$ and $1$ to 
the internal vertices of $P$. A \emph{($v_0$-)balloon} is a pair $B=(P,C)$ where $C \in \cC$
and $P$ is a $v_0$-$C$-path; the \emph{weight} of $B$ is $w(B)=2w(P)+w(C)$. 
We call the endpoint $v_\ell$ of $P$ the \emph{knot vertex} of $B$. 
We let $V(B)$ (respectively $V(C)$, $V(P)$) denote the set of vertices occurring in $B$
(respectively in $C$, in $P$); hence $V(C) \cap V(P) = \{v_\ell\}$ is the knot vertex. 
The \emph{edges used in $B$}, $E(B)$, is the set of edges required for $B$
to be a balloon, i.e., the edges $v_iv_{i+1}$ of $P$ and the edges of $C$.
Note that this does not necessarily include all edges of $G[V(B)]$.
The \emph{edges used in $P$} (in $C$) is defined correspondingly.

Having fixed $v_0$ and $\cC$ as above, we define a polytope $\cP$ by constraints
\begin{equation}
  \label{eq:balloonconstraint}
  w(B) \geq 1 \textrm{ for every $v_0$-balloon $B=(P,C)$}.
\end{equation}
with $x_v \geq 0$ for every vertex $v$ and $x_{v_0}=0$.
We refer to this as the \emph{local LP}.

Given an optimisation goal $\min c^T x$ for the above LP, 
the dual of the system asks to pack balloons, at weight $1$ for every
balloon, subject to every vertex $v$ having a capacity $c_v$ and a balloon
$B=(P,C)$ using capacity from $v$ in proportion to the coefficient of $v$
in $w(B)$ (which is 2 if $v \in V(P)$, and $1$ otherwise).

\subsection{Balloons and path pairs}

We now give some observations that will simplify the future arguments
regarding the local LP, and in particular will allow us to change
perspectives between viewing the constraints as pairs of paths, or as
rooted unbalanced cycles.

Let $x: V \to [0,1]$ be a fractional assignment. We first observe that our
weights $w(P)$ and $w(C)$ for balloons $B=(P,C)$ can be recast as
\emph{edge weights}: For any edge $uv \in E$, we let the length of $uv$ 
under $x$ be $\ell_x(uv)=(x_u+x_v)/2$. 
For a path $P$, we let $\ell_x(P)=\sum_{uv \in E(P)} \ell_x(uv)$
be the length of $P$ under this metric, and similarly for simple cycles.
We also define $z_x(v)=\min_P \ell_x(P)$ ranging over all $v_0$-$v$-paths
$P$ as the \emph{distance to $v$ (under $x$, from $v_0$)}. This metric 
agrees well with the notion of the weight of a balloon that we use in the LP, 
as we will see. Note that the end points of a path $P$ contribute only half
their weight to the length $\ell_x(P)$ of a path, as in $w(P)$.

Now observe that for any balloon $B=(P,C)$, and any vertex $v \in C$
other than the knot vertex, it is possible to form two paths 
$P_1, P_2$ from $v_0$ to $v$, such that $P_1 \cup P_2$ covers $B$,
and such that the weight of the balloon equals $\ell_x(P_1)+\ell_x(P_2)$.
Indeed, this sum gives coefficient $2$ to every vertex of $P$, and
coefficient $1$ to every vertex of $V(C) \setminus V(P)$, including
the vertex $v$. We refer to this as a path decomposition of $B$. 
We also observe two alternative decompositions. 

\begin{lemma} \label{lemma:decompositions}
  Let $B=(P,C)$ be a $v_0$-balloon. Each of the following is an equivalent
  decomposition of the weight of $B$.
  \begin{enumerate}
  \item $w(B) \geq \ell_x(C) + \min_{v \in V(C)} 2z_x(v)$,
    with equality achieved if $B$ is a min-weight $v_0$-balloon.
  \item $w(B)=\ell_x(P_1)+\ell_x(P_2)$, for any decomposition of $B$
    into two paths $P_1, P_2$ ending at a non-knot vertex $v \in V(C)$.
  \item  \label{decompose:uvpaths}
    $w(B)=\ell_x(P_u)+\ell_x(P_v)+\ell_x(uv)$ 
    for any decomposition
    of $B$ into one path $P_u$ to $u \in V(C)$, one path $P_v$ to $v \in V(C)$,
    and an edge $uv \in E(C)$, where $u, v \in V(C)$ are non-knot vertices.
    Equivalently, $w(B)=\sum_{w \in P_u} x_w + \sum_{w \in P_v} x_w$. 
  \end{enumerate}
\end{lemma}

Note that the first decomposition here implies that the knot vertex $v$
of a min-weight balloon $B$ will be chosen for minimum $z_x$-value.

In general, we view constraints as pairs of paths when deriving simple
properties of the LP, but will need to revert to the view of biased graphs
when arguing persistence in the next section.

\subsection{Structure of min-weight balloons}
\label{sec:minwtballoon}
We now work closer towards a separation oracle, by showing properties of
balloons $B$ minimising $w(B)$ under an assignment $x$, as this will help us
finding the \emph{most violated constraint} of an instance of the local LP.

\begin{figure}\centering
  \begin{subfigure}[b]{0.25\textwidth}\centering
    \begin{tikzpicture}
      \node[circle,draw,thick,dashed, minimum width=2 cm] (c) at (0,3) {};
      \node (v0) at (0,0) {$\bullet$}; \node[anchor=west] at (v0) {$v_0$};
      \node (vk) at (0,2) {$\bullet$}; \node[anchor=north west] at (vk) {$v_k$};
      \node (u) at (0,1) {$\bullet$}; \node[anchor=west] at (u) {$u$};
      \node (v) at (-1,3) {$\bullet$}; \node[anchor=east] at (v) {$v$};
      
      \path[draw,thick] (v0) -- (vk) arc (270:180:1cm);
      \path[draw,thick] (u.center) edge[bend left] node [left] {} (v.center);
    \end{tikzpicture}
    \caption{$u$ on $P$, case 1}
  \end{subfigure}%
  \begin{subfigure}[b]{0.25\textwidth}\centering
    \begin{tikzpicture}
      \node[circle,draw,thick,dashed, minimum width=2 cm] (c) at (0,3) {};
      \node (v0) at (0,0) {$\bullet$}; \node[anchor=west] at (v0) {$v_0$};
      \node (vk) at (0,2) {$\bullet$}; \node[anchor=north west] at (vk) {$v_k$};
      \node (u) at (0,1) {$\bullet$}; \node[anchor=west] at (u) {$u$};
      \node (v) at (-1,3) {$\bullet$}; \node[anchor=east] at (v) {$v$};
      
      \path[draw,thick] (v0) -- (vk) arc (270:540:1cm);
      \path[draw,thick] (u.center) edge[bend left] node [left] {} (v.center);
    \end{tikzpicture}
    \caption{$u$ on $P$, case 2}
  \end{subfigure}%
  \begin{subfigure}[b]{0.25\textwidth}    \centering
    \begin{tikzpicture}
      \node[circle,draw,thick,dashed, minimum width=2 cm] (c) at (0,2) {};
      \node (v0) at (0,0) {$\bullet$}; \node[anchor=west] at (v0) {$v_0$};
      \node (vk) at (0,1) {$\bullet$}; \node[anchor=north west] at (vk) {$v_k$};
      \node (u) at (-1,2) {$\bullet$}; \node[anchor=east] at (u) {$u$};
      \node (v) at (0,3) {$\bullet$}; \node[anchor=south] at (v) {$v$};

      \path[draw,thick] (u.center) -- (v.center);
      \path[draw,thick] (v0) -- (vk) arc (270:180:1cm);
      \path[draw,thick] (vk) arc (270:450:1cm);
    \end{tikzpicture}
    \caption{$u$ on $C \setminus P$, case 1}
  \end{subfigure}%
  \begin{subfigure}[b]{0.25\textwidth}    \centering
    \begin{tikzpicture}
      \node[circle,draw,thick,dashed, minimum width=2 cm] (c) at (0,2) {};
      \node (v0) at (0,0) {$\bullet$}; \node[anchor=west] at (v0) {$v_0$};
      \node (vk) at (0,1) {$\bullet$}; \node[anchor=north west] at (vk) {$v_k$};
      \node (u) at (-1,2) {$\bullet$}; \node[anchor=east] at (u) {$u$};
      \node (v) at (0,3) {$\bullet$}; \node[anchor=south] at (v) {$v$};
      
      \path[draw,thick] (u.center) -- (v.center);
      \path[draw,thick] (v0) -- (vk) arc (270:90:1cm);
    \end{tikzpicture}
    \caption{$u$ on $C \setminus P$, case 2}
  \end{subfigure}  
  \caption{Reconfigurations in Lemma~\ref{lemma:balloon-shortest-paths}. Note that in every case, the new balloon~$B'$
    decomposes into the new shortest path $P_v$ and a pre-existing
    path from $v_0$ to $v$ in~$B$.}
  \label{fig:shortest}
\end{figure}

Our first lemma is a structural result that will be independently useful
in the next section. 

\begin{lemma} \label{lemma:balloon-shortest-paths}
  Let $B=(P,C)$ be a min-weight balloon with respect to an assignment
  $x: V \to [0,1]$. Then for every vertex $v$ in $B$,
  $B$ contains a shortest $v_0$-$v$-path under the metric $\ell_x$. 
\end{lemma}
\begin{proof}
  Clearly, $P$ must be a shortest path, by decomposition 1 of
  Lemma~\ref{lemma:decompositions}, hence the claim holds for every
  vertex of $P$. 
  Let $v \in C \setminus P$, and let $P_v$ be a shortest $v_0$-$v$-path such that
  both paths to $v$ along $B$ (as in decomposition 2) are longer than $P_v$;
  let $v$ be chosen with minimum $z_x(v)$-value, subject to these conditions.
  Choose $P_v$ to maximize the length of the prefix of $P_v$ that follows $B$;
  in particular, $P_v \cap P$ is a common prefix of $P$ and $P_v$.
  We may also assume that $P_v$ ``rejoins'' $B$
  only exactly once, i.e., after $P_v$ has followed its first edge not
  present in $B$, then no further vertex of $P_v$ other than $v$
  is contained in $B$.  Indeed, let $w$ be the first vertex of $B$ that $P_v$ 
  encounters after having departed from $B$.  Then by assumption the prefix 
  of $P_v$ up to $w$ is shorter than both paths to $w$ along $B$, 
  and instead of $(v, P_v)$ above we may choose $(w, P_w)$ where $P_w$ is 
  the aforementioned prefix of $P_v$. Note that $P_w$ now has all the
  properties assumed of $P_v$.
  Let $u$ denote the departure point of $P_v$ from $B$
  (i.e., $u$ is the last vertex of $P_v$ such that the prefix of $P_v$ 
  from $v_0$ to $u$ follows edges of $B$). Let $v_k$ be the knot vertex of $B$.
  We split into two cases. See Figure~\ref{fig:shortest} for an
  illustration. 

  If $u$ lies in $P$, then we reason as follows. 
  The two paths from $u$ to $v_k$ (in $P$) respectively $v$ (in $P_v$) 
  form a chord path for $C$. Let $C'$ be a reconfiguration of $C$ by
  this path, and use $P'=P \cap P_v$ as path to attach $C'$ to $v_0$,
  defining a balloon $B'=(P',C')$. Decompose $B$ as $P_1+P_2$, where 
  $P_1$ and $P_2$ both end in $v$, and similarly decompose $B'$
  as $P_1'+P_2'$, ending in $v$. Observe that since $u$ is the new
  knot vertex, for both possible choices of $C'$ it holds that
  one of the paths $P_i'$ will be identical to $P_v$,
  while the other will be either $P_1$ or $P_2$.
  Since $z(P_v) < z(P_1), z(P_2)$, the new balloon $B'$ represents
  a constraint with a smaller value than $B$.

  Otherwise, the departure point $u$ lies in $C \setminus P$. 
  The suffix of $P_v$ from $u$ to $v$ forms a chord path; reconfiguring
  by this chord path leaves two options for the new cycle $C'$, 
  with $C'$ either including $v_k$ or excluding $v_k$. 
  If $C'$ includes $v_k$, then we may
  choose $v_k$ as our new knot vertex. As a result, we may use the
  same argument as in the previous paragraph, noting that one of the 
  two paths in the decomposition of $B'$ was also present in the 
  decomposition of $B$, whereas the other path is shorter than both
  previous paths. Otherwise, finally, the knot vertex will be $u$,
  and again one of the two paths in the decomposition of $B'$
  is present in the decomposition in $B$, while the other is the
  new shortest path to $v$. 
\maybeqed
\end{proof}

We observe a particular consequence useful for finding min-weight balloons. 

\begin{corollary} \label{cor:balloons-twoshortest}
  Any minimum-weight balloon constraint $B=(P,C)$ can be decomposed as 
  $w(B)=z_x(u)+z_x(v)+(x(u)+x(v))/2$, for some $u, v \in V(C)$ with
  $uv \in E$. 
\end{corollary}
\begin{proof}  
  Let $B=(P,C)$ be a minimum-weight balloon. Recall that $|V(C)|\geq 3$
  since the input graph is simple (by Prop.~\ref{prop:simplegraph}).
  If there is a vertex 
  $u \in V(C) \setminus V(P)$ such that both paths to $u$ in $B$ 
  are shortest paths under the metric $\ell_x$, then the result follows.
  Indeed, decompose $B$ into two paths $P_1$, $P_2$ ending at $u$,
  and truncate $P_2$ by one edge to end at some vertex $v \neq u$,
  where $v$ is chosen not to be the knot vertex.
  Then this creates shortest paths to $u$ and to $v$, 
  and by Lemma~\ref{lemma:decompositions}, decomposition~\ref{decompose:uvpaths}
  these paths and the edge $uv$ form a decomposition of $B$.
  
  Otherwise, every vertex in $V(C) \setminus V(P)$ has exactly one
  shortest path in $B$ under $\ell_x$, by Lemma~\ref{lemma:balloon-shortest-paths}.
  There is then (precisely) one edge $uv \in E(C)$ that is not used by any 
  shortest path in $B$, and neither of $u$ or $v$ is the knot vertex. 
  In this scenario, again, the shortest paths in $B$ to $u$ and to $v$
  and the edge $uv$
  form a decomposition of $B$ as in Lemma~\ref{lemma:decompositions}, decomposition~\ref{decompose:uvpaths}
  and we are done. 
%
\maybeqed
\end{proof}

\subsection{The separation oracle}
\label{sec:seporacle}

We now finish the results of this section by constructing a
polynomial-time separation oracle over the local LP, assuming a membership
oracle for the class $\cB$. The result essentially follows from
Corollary~\ref{cor:balloons-twoshortest} by slightly perturbing 
edge lengths so that shortest paths are unique.

\begin{theorem} \label{thm:optimise}
  Let $\Psi=(G=(V,E),\cB)$ be a biased graph, and let $v_0 \in V$
  be the root vertex. Assume that we have access to a polynomial-time
  membership oracle for $\cB$, that for every simple cycle $C$ of $G$
  can inform us whether $C \in \cB$ or not. 
  Then there is a polynomial-time separation oracle for the 
  local LP rooted in $v_0$.
\end{theorem}
\begin{proof}
  Let $x: V \to [0,1]$ be a fractional assignment, where we want to 
  decide whether $x$ is feasible for the LP, i.e., we wish to decide
  whether there is a $v_0$-balloon $B$ such that $w(B)<1$ under $x$.
  Order the edges of $G$ as $E=\{e_1,\ldots,e_m\}$. We will
  associate with each edge $e_i \in E$ a tuple $(\ell_x(e_i), 2^i)$, 
  and modify our distance measure to work componentwise over
  these tuples, and order distances over such tuples
  in lexicographical order. This is to deterministically 
  simulate the perturbation of all weights by a
  random infinitesimal amount.
  Note that all paths and cycles (in fact, all sets of edges) have unique 
  lengths in this way. In particular, for every vertex $u$
  there is a unique shortest path $P_u$ from $v_0$ to $u$.
  Also observe that computations over these weights (and the computations
  of shortest paths) can be done in polynomial time.

  Assume that there is a balloon $B$ with $w(B)<1$. Then the modified
  weight of $B$ will be $(w(B), \alpha)$ for some $\alpha \in \N$,
  where clearly $w(B)<1$ still holds. Also note that 
  Corollary~\ref{cor:balloons-twoshortest} still applies to the modified weights,
  since the proof consists only of additions and comparisons. 
  Hence, if $B$ is the (unique) min-weight balloon under the 
  modified weights, then $B$ can be decomposed into two shortest paths 
  $P_u$, $P_v$ and an edge $uv$, so that
  $w(B)=\ell_x(P_u)+\ell_x(P_v)+\ell_x(uv)=z_x(u)+z_x(v)+\ell_x(uv)$.
  Since shortest paths are unique, we can find the components $P_u$,
  $P_v$ and $uv$ defining $B$, if starting from the edge $uv$.
  From these paths, it is easy to reconstruct the cycle $C$
  and verify via an oracle query that $C \notin \cB$. 
  By iterating the procedure over all edges $uv \in E$, 
  we will find the min-weight balloon. 
\maybeqed
\end{proof}

\section{Half-integrality and persistence of the local LP}

We now proceed to use the insights gained in the previous section
to show that the local LP is in fact half-integral and obeys a 
strong persistence property.
The approach of the proof is closely related to that of Guillemot
for variants of \textsc{Multiway Cut}~\cite{Guillemot11}.

\subsection{Half-integrality}
\label{sec:halfint}

Let $c: V \to \Q$ be arbitrary vertex weights, 
let $x^*$ be an optimal solution to the above LP, 
and let $y^*$ be an optimum solution to the dual.
By complementary slackness, if $y^*_B>0$ then $w(B)=1$ under $x^*$,
and if $x^*_v>0$ then the packing $y^*$ saturates $v$ to capacity $c_v$. 
Let $V_R$ denote the set of vertices reachable from $v_0$ at distance $0$
(also implying that $x^*_v=0$ for every $v \in V_R$). We let $V_1=\{v \in V: x_v=1\}$
and $V_{1/2}=N(V_R) \setminus V_1$. We claim that $\tfrac 1 2 V_{1/2}+V_1$ is a new
LP-optimum. This will follow relatively easily from the ``path pair perspective'' 
on balloons.

First, we give a structural lemma.

\begin{lemma} \label{lemma:balloonsupporttypes}
  Let $B=(P,C)$ be a balloon with $w(B)=1$ with respect to the
  vertex weights $x^*$. 
  Then $B$ is of one of the following types:
  \begin{enumerate}
  \item $B$ is contained within $G[V_R+v]$ for some vertex $v$, 
    where $v \in V_1 \cap (V(C) \setminus V(P))$;
  \item $B \cap N(V_R)$ is a single vertex $v \in V_{1/2}$, and $v \in V(P)$;
  \item $B \cap N(V_R)$ consists of two vertices $u,v$ which cut $C$ ``in
    half'' as follows: $C$ consists of two $uv$-paths of which one is
    contained in $V_R$ and contains at least one internal vertex $v'$
    which is the knot vertex of $B$, and the other path is disjoint from $V_R$. 
  \end{enumerate}
\end{lemma}
\begin{proof}
  Since $x^*$ is a feasible solution for the local LP, 
  any balloon with $w(B)=1$ is a min-weight balloon. 
  Thus by Corollary~\ref{cor:balloons-twoshortest}, we can 
  decompose $B$ as $P_u+P_v+uv$ where $P_u$ and $P_v$ are shortest paths.
  In particular, both $P_u$ and $P_v$ contain a prefix in $V_R$, 
  make at most one visit to $N(V_R)$, and proceed to subsequently 
  not revisit $N[V_R]$ at all.
  Also note that $B$ necessarily intersects $N(V_R)$, 
  since $B$ is a connected subgraph rooted in $v_0$
  but not contained in $V_R$. 

  First assume that $B$ intersects a vertex of $v_1 \in V_1$. Then since $w(B)=1$,
  only one of the paths $P_u$, $P_v$ contains $v_1$, whereas the other path
  is entirely contained in $V_R$. But then we must have $v_1 \in \{u,v\}$,
  as otherwise the edge $uv$ cannot exist. We conclude that in this case,
  $B$ is a balloon of type 1.

  Next, assume that $B$ intersects $N(V_R)$ in a single vertex $v' \in V_{1/2}$.
  We claim that $B$ must intersect $v'$ with a coefficient of $2$:
  Indeed, if not, then one of the paths $P_u$ and $P_v$, say $P_u$,
  must be entirely contained in $V_R$, in which case $v \in N(u)$
  must be contained in $N(V_R)$ and we are back in the previous case
  (contradicting $v' \in V_{1/2}$). Thus $v' \in V(P)$, and $B$ is a 
  balloon of type 2.

  Finally, assume that $B$ intersects $N(V_R)$ in at least two vertices.
  Then in fact $B$ intersects $N(V_R)$ in exactly two vertices $u', v'$,
  by the properties of $P_u$ and $P_v$, and neither of these vertices are
  in $V_1$, since $w(B)=1$. Furthermore, since each of $P_u$ and $P_v$
  intersects $N(V_R)$ only once, these vertices $u'$, $v'$ lie 
  after the common part of $P_u$ and $P_v$, i.e., in $V(C) \setminus V(P)$.
  Then indeed the knot vertex lies in $V_R$, whereas the path
  from $u'$ to $v'$ via $uv$ lies outside of $V_R$, as required. 
\maybeqed
\end{proof}

We now show that complementary slackness implies that the new
fractional assignment is actually an LP-optimum.

\begin{lemma} \label{lemma:halfoptimum}
  The assignment $V_1+\tfrac 1 2 V_{1/2}$ is an LP-optimum for the local LP.
\end{lemma}
\begin{proof}
  We first show that $V_1+\tfrac 1 2 V_{1/2}$ is a valid LP-solution.
  Assume towards a contradiction that $w(B)<1$ for some balloon $B=(P,C)$
  under the proposed weights. Note that $V_1 \cup V_{1/2}$ intersects
  every balloon since $G[V_R]$ is balanced. If $w(B)<1$ we must thus have
  $B \cap (V_1 \cup V_{1/2})=\{v\}$ for some $v \in V_{1/2}$, where $v \in
  (V(C) \setminus V(P))$. But then $V(B) \subseteq V_R \cup \{v\}$ 
  (since otherwise $C$ passes the ``border'' $N(V_R)$ in at least two locations), 
  which implies $x^*_v=1$, contrary to assumptions.

  We now show optimality. 
  By complementary slackness, $y^*$ is a packing of balloons saturating
  every $v \in N(V_R)$ to its capacity $c_v$, with $w(B)=1$ for every
  balloon $B$ in the support of $y^*$; hence $B$ will be of one of the
  three types of Lemma~\ref{lemma:balloonsupporttypes}. For $i=1, 2, 3$,
  let $\cB_i$ be the set of balloons of type $i$ from the support of
  $y^*$. By optimality, $c^Tx^*=\sum_B y^*_B$. Note that a balloon
  of type 1 intersects one vertex in $V_1$ with coefficient 1 and no other
  vertex in $V_1 \cup V_{1/2}$, while a balloon of type 2 or 3 intersects
  one vertex in $V_{1/2}$ with coefficient 2, respectively two vertices in
  $V_{1/2}$ with coefficient 1 each, and no other vertex from $V_1 \cup
  V_{1/2}$.  Since every vertex of $V_1 \cup V_{1/2}$ is in the support of $x^*$,
  these vertices are saturated by the packing $y_B$. For a balloon $B$ and vertex $v$, 
  let $c(B,v) \in \{0,1,2\}$ be the coefficient of $v$ in the constraint $w(B) \geq 1$. 
  We thus get 
  \[
  \sum_{v \in V_1} c_v = \sum_{v \in V_1} \sum_{B \in \cB_1} c(B,v) \cdot y^*_B
  = \sum_{B \in \cB_1} y^*_B \sum_{v \in V_1} c(B,v) = \sum_{B \in \cB_1} y^*_B
  \]
  and
  \[
  \sum_{v \in V_{1/2}} c_v = \sum_{v \in V_{1/2}} \sum_{B \in \cB_2 \cup \cB_3} c(B,v) \cdot y^*_B
  = \sum_{B \in \cB_2 \cup \cB_3} y^*_B \sum_{v \in V_{1/2}} c(B,v) = \sum_{B \in \cB_2 \cup \cB_3} 2y^*_B,
  \]
  hence
  \[
  c^Tx^*=\sum_B y^*_B = c^T(V_1 + \frac 1 2 V_{1/2})
  \]  
  which shows that $V_1+\tfrac 1 2 V_{1/2}$ is an LP-optimum. 
\maybeqed
\end{proof}

\subsection{Persistence and the tiling property}
\label{sec:tiling}
Finally, we reach the statement concerning the persistence of the local LP. 
Since the statement is somewhat intricate, let us walk through it. First of all, 
the basic persistence property is similar to that used in 
Multiway Cut~\cite{GargVY04,Guillemot11,CyganPPW13MWC}.
Let $V_R$ be defined as before, and for a solution $X \subseteq V$ to \localproblem
(hence $v_0 \notin X$) 
define $S_X$ as  the set of vertices of the connected component of $G-X$ containing $v_0$.
Then persistence dictates that there is an optimal solution $X$
such that if $z_x(v)=0$, i.e., 
if $v \in V_R$ then $v \in S_X$, and if $v \in V_1$ then $v \in X$. 
We note that both these properties hold for the computed set $S'$ in the below lemma. 

However, the lemma also gives a useful tiling property, for the purposes
of solving the global problem: Let $X$ be a solution to the full, global
\problem problem, with $v_0 \notin X$, and let $S$ be the vertices
reachable from $v_0$ in $G-X$. Then it ``does not hurt'' the global solution
to assume that the induced solution $X \cap N(S)$ to the local problem
also observes the persistence properties as above. 
This is implied by the closed-neighbourhood condition on $S^+$:
We will ``cut away'' a section $G[S^+]$ of the initial graph,
and find a new solution for it such that all vertices of $S^+$
neighbouring $V \setminus S^+$ are deleted, and so that the solution
in $G[S^+]$ respects persistence properties while not being more expensive
than the original solution $X \cap S^+$.
See Figure~\ref{fig:reconfigure} for an illustration of the modification of $S$.
This allows us to assemble a solution to the global problem
out of pieces computed for instances of the local problem.

\begin{figure}
  \begin{subfigure}{0.5\textwidth}\centering
  \scalebox{0.85}{\begin{tikzpicture}
  \node[circle, draw, fill=green!50, minimum width = 6.3 cm,thick] (sbig) at (3,0) {};
  \node[circle,draw,minimum width=5 cm,fill=white] (s) at (3,0) {};
  \node[circle split, draw, circle split part fill={red!50, blue!50}, minimum width = 6 cm,thick] (vrbig) at (0,0) {};
  \draw[thick] (0,0) -- (30:3 cm);
  \draw[thick] (0,0) -- (-30:3 cm);
  \node[circle,draw, minimum width=5 cm,fill=white] (vr) at (0,0) {$V_R$};

  \node[circle, draw, dashed, minimum width = 6.3 cm] at (3,0) {};

  \node[circle,draw,minimum width=5 cm] at (3,0) {};

  \node[anchor=south] at (vr.north) {$V_1$};
  \node[anchor=north] at (vr.south) {$V_{1/2}$};
  
  \node[anchor=south] at (s.north) {$N(S) \setminus N[V_R] $};

  \node at (3.5,0) {$S$};
  \node[circle,draw,minimum width=1.5 cm] (uset) at (1.5,0) {};

  \node[anchor=south] at (uset.south) {$U$};

  \node (v0) at (1.5,0) {$v_0$};

  \node (va) at (1.8,0.25)  {$\bullet$};
  \node (vb) at (1.8,-0.25)  {$\bullet$};
  
  \node (vc) at (2.6,0.75) {$\bullet$};
  \node (vd) at (2.6,-0.75) {$\bullet$};
  
  \path[draw] (va) -- (v0) -- (vb);
  \path[draw] (vc.center) -- (va.center);
  \path[draw] (vd.center) -- (vb.center);

  \node(vl) at (1.1,0) {$\bullet$};
  \node(vll) at (0.3,0.5) {$\bullet$};
  \node(vlh) at (0.3,-0.5) {$\bullet$};

  \path[draw] (vlh.center) -- (vl.center) -- (vll.center);
  \path[draw] (vl.center) -- (1.35,0);

  \node[circle,draw,minimum width=0.75 cm] (utop) at (1.5,1.2) {$\bullet$};
  \node[circle,draw,minimum width=0.75 cm] (ubot) at (1.5,-1.2) {$\bullet$};
  
  \node (u1) at (2,1.8) {$\bullet$};
  \node (u2) at (2,-1.8) {$\bullet$};
  \node[right=of u1] (u3) {$\bullet$};
  \node[right=of u2] (u4) {$\bullet$};
  
  \path[draw] (vll.center) -- (utop.center) -- (u1.center) -- (u3.center) -- (vc.center);
  \path[draw] (vlh.center) -- (ubot.center) -- (u2.center) -- (u4.center) -- (vd.center);
  
\end{tikzpicture}}
\caption{The sets $S$ and $N[S]$ interacting with~$V_R$.}
\end{subfigure}
\begin{subfigure}{0.5\textwidth}\centering
\scalebox{0.85}{\begin{tikzpicture}
  \node[circle, draw, fill=green!50, minimum width = 6.3 cm,thick] (sbig) at (3,0) {};
  \node[circle,draw,minimum width=5 cm,fill=white] (s) at (3,0) {};
  \node[circle split,draw,fill=green!50,minimum width=6 cm] at (0,0) {};
  \path[draw,fill=white]  (0,0) -- (3,0) arc (0:-30:3cm)--cycle;
  \draw[thick] (0,0) -- (30:3 cm);
  \draw[thick] (0,0) -- (-30:3 cm);
  \node[circle,draw, minimum width=5 cm,fill=white] (vr) at (0,0) {$V_R$};


  \node[anchor=south] at (vr.north) {$V_1$};
  \node[anchor=north] at (vr.south) {$V_{1/2}$};
  
  \node[anchor=south] at (s.north) {$N(S) \setminus N[V_R] $};

  \node at (4,0.5) {$S \setminus N[V_R]$};

  \node[draw,dashed,circle,minimum width=1.5 cm] (uset) at (1.5,0) {};

  \node[anchor=south] at (uset.south) {$U$};

  \node (v0) at (1.5,0) {$v_0$};

  \node (va) at (1.8,0.25)  {$\bullet$};
  \node (vb) at (1.8,-0.25)  {$\bullet$};
  
  \node (vc) at (2.6,0.75) {$\bullet$};
  \node (vd) at (2.6,-0.75) {$\bullet$};
  
  \path[draw] (va) -- (v0) -- (vb);
  \path[draw] (vc.center) -- (va.center);
  \path[draw] (vd.center) -- (vb.center);

  \node(vl) at (1.1,0) {$\bullet$};
  \node(vll) at (0.3,0.5) {$\bullet$};
  \node(vlh) at (0.3,-0.5) {$\bullet$};

  \path[draw] (vlh.center) -- (vl.center) -- (vll.center);
  \path[draw] (vl.center) -- (1.35,0);

  \node (utop) at (1.5,1.2) {$\bullet$};
  \node (ubot) at (1.5,-1.2) {$\bullet$};
  
  \node (u1) at (2,1.8) {$\bullet$};
  \node (u2) at (2,-1.8) {$\bullet$};
  \node[right=of u1] (u3) {$\bullet$};
  \node[right=of u2] (u4) {$\bullet$};
  \node[anchor=west] at (vd.north west) {$N(U) \cap V_{1/2} \cap S$};
  
  \path[draw] (vll.center) -- (utop.center) -- (u1.center) -- (u3.center) -- (vc.center);
  \path[draw] (vlh.center) -- (ubot.center) -- (u2.center) -- (u4.center) -- (vd.center);
\end{tikzpicture}}
\caption{The sets $S'$ (white) and $S^+$ (green+white).} 
\end{subfigure}
\caption{Illustration of the reconfiguration in the tiling property in
  Lemma~\ref{lemma:persistence}.  The sets $S$ and $N(S)$ are reconfigured using $V_R$, $V_1$, $V_{1/2}$ 
into $S'=V_R \cup (V_{1/2} \cap N(U) \cap S) \cup (S \setminus N[V_R])$ (white in second picture) and $S^+ \setminus S'$ (green in second picture)}
\label{fig:reconfigure}
\end{figure}

\begin{lemma} \label{lemma:persistence}
  Let $x=V_1+\tfrac 1 2 V_{1/2}$ be the half-integral optimum from above,
  and let $V_R$ be the corresponding reachable region. Let $S$ be a
  balanced set with $v_0 \in S$. Then we can grow the closed region $N[S]$
  to $N[S \cup V_R]$ without paying a larger cost for deleting vertices.
  More formally, there is a set of vertices $S^+$ and a set $S' \subseteq S^+$ 
  such that $G[S']$ is balanced and the following hold.
  \begin{enumerate}
  \item $S^+=N[S \cup V_R]$;
  \item $N[S'] \subseteq S^+$;
  \item $V_R \subseteq S'$;
  \item $V_1 \subseteq (S^+ \setminus S')$;
  \item $c(S^+ \setminus S') \leq c(N(S))$.
  \end{enumerate}
\end{lemma}
\begin{proof}
    Let $U$ be the connected component of $v_0$ in $G[S \cap V_R]$.
  We define the sets $S^+=N[S \cup V_R]$, and
  \[
  S' = V_R \cup (N(U) \cap V_{1/2} \cap S) \cup (S \setminus N[V_R]).
  \]
  Observe that $S' \subseteq S \cup V_R$ and that $V_1 \subseteq N(S')$. 
  We first show that $G[S']$ is balanced. Assume not, and let
  $C$ be an unbalanced cycle contained in $G[S']$. We may assume
  that $C$ is reduced as by Lemma~\ref{lemma:cycleintersect}
  with respect to $V_R$. Observe that $C$ must intersect $V_R$,
  as otherwise $V(C) \subseteq S$ contradicting that $G[S]$ is balanced.
  By Lemma~\ref{lemma:cycleintersect}, this intersection takes
  the form of a simple path $P_{ab}$ connecting two vertices $a, b \in N(V_R)$.
  Furthermore, we have $a, b \in N(U) \cap V_{1/2} \cap S$, 
  and the path $P_{ab}$ intersects $V_R \setminus S$, i.e.,
  $P_{ab}$ contains internal vertices not contained in $U$.
  Let $P_a$ respectively $P_b$ be the prefix respectively suffix of $P_{ab}$
  contained in $U$, if any, together with the vertices $a$ respectively $b$.
  Let $P_{ab}'$ be a new chord path connecting $P_a$ and $P_b$ in $U$,
  and let $C'$ be a new cycle resulting from the reconfiguration of $C$ by $P_{ab}'$.  
  Since $C'$ cannot be contained in $S$, this cycle must be formed from $P_{ab}+P_{ab}'$,
  and since $C'$ cannot be contained in $V_R$ it must still contain the vertices $a$ and $b$ 
  (i.e., we had $P_a=a$ and $P_b=b$; recall that $a,b \in V_{1/2}$). But then $V(C') \cap V_R$ consists
  of two distinct paths; use a chord path $P''$ between these paths
  to reconfigure $C'$ into a new cycle $C''$. Then we may see that 
  $|C'' \cap N(V_R)|=1$, namely only of the vertices $\{a,b\}$, 
  contradicting that $a, b \in V_{1/2}$. We conclude that $S'$ is balanced. 

  Items 1--4 in the lemma hold by definition or are easy, so it remains to show
  that $c(S^+ \setminus S') \leq c(N(S))$. Let us break down this expression.
  Writing $S^+=V_R \cup S \cup N(V_R) \cup N(S)$, first note that 
  $S^+ \setminus S' \subseteq N(V_R) \cup N(S)$ by definition of $S'$;
  more carefully,
  \[
  S^+ \setminus S' = (N(S) \setminus S') \cup (N(V_R) \setminus S').
  \]
  Vertices of $N(S) \setminus S'$ contribute equally to both sides of the inequality
  and can be ignored, hence we are left with 
  vertices of $N(V_R) \setminus (S' \cup N(S))$ contributing to the left hand side
  and vertices of $N(S) \cap S'=N(S) \cap V_R$ contributing to the right hand side. 
  Splitting $N(V_R)=V_1 \cup V_{1/2}$, the former set simplifies to
  $
  (V_1 \setminus N(S)) \cup (V_{1/2} \setminus (N(U) \cup N(S)).
  $
  Relaxing slightly we define
  \[
  Z:=(V_1 \setminus N(S)) \cup (V_{1/2} \setminus N(U))
  \]
  and 
  \[
  Y:=N(U) \cap V_R;
  \]
  it will suffice to show $c(Z) \leq c(Y)$. 
  This will occupy the rest of the proof.

  Let $y^*$ be the dual optimum, i.e., a fractional packing of balloons
  which saturates $v$ for every $v \in Z$, with each balloon $B$ in the
  support being of types 1--3 of Lemma~\ref{lemma:balloonsupporttypes} (by
  complementary slackness). Note that every vertex of $Z$ is in the support of $x$. 
  Let $\cB_1$ contain the balloons from the
  support of $y^*$ which intersect $Z$ with a total coefficient of $1$
  (i.e., $\cB_1$ contains balloons of type 1, and balloons of type 3 which
  intersect $Z$ in only one vertex), and let $\cB_2$ contain those which
  intersect $Z$ with a total coefficient of $2$ (i.e., balloons of type 2,
  and balloons of type 3 which intersect $Z$ in two vertices). Note that no
  balloon from the support of $y^*$ intersects $Z \subseteq N(V_R)$ with 
  a total coefficient of more than 2. Then
  \[
  c(Z) = \sum_{B \in \cB_1} y^*_B + \sum_{B \in \cB_2} 2y^*_B.
  \]
  We need to show that every $B \in \cB_1$ intersects $Y$ with a total
  coefficient of at least 1, and every $B \in \cB_2$ intersects $Y$ with a
  total coefficient of at least 2. The inequality will follow. 

  First consider $v \in V_1 \cap Z$, and let $B \in \cB_1$ intersect $v$.
  Then $B$ is of type 1, hence contained in $G[V_R+v]$. If $B-v \subseteq U$,
  then $B \cap N(S)=\{v\}$, but $v \in V_1 \cap Z$ implies $v \notin N(S)$;
  hence not all of $B-v$ is contained in $U$, and $B$ intersects $Y$. 
  
  Next, consider a vertex $v  \in V_{1/2} \cap Z$ and a balloon $B \in \cB_1$
  intersecting $v$ with coefficient 1; hence $B$ is of type 3. 
  Since $v \in Z$ we have $v \notin N(U)$; since $B$ 
  connects $v_0 \in U$ with $v \notin N[U]$ using internal vertices in $V_R$, 
  there must exist some vertex $u \in B$ contained in $N(U) \cap V_R=Y$. 
  Hence $B$ intersects $Y$. 

  Thirdly, consider some $v \in V_{1/2} \cap Z$ intersecting some $B \in \cB_2$
  with a coefficient of 2. Again $v \notin N(U)$, and $B$ contains
  a path from $v_0$ to $v$ with internal vertices in $V_R$; furthermore
  every vertex on this path has coefficient $2$ in $B$. Thus $B$ 
  intersects $Y$ with a coefficient of at least 2. 

  Finally, consider a balloon $B \in \cB_2$ of type 3, intersecting
  two vertices $v, v' \in V_{1/2} \cap Z$. Then $B$ traces two paths
  from $v_0$ to $v, v'$, and either both these paths intersect $Y$
  in a single vertex (which then has coefficient 2), or two distinct
  vertices of $Y$ intersect $B$ (for a total coefficient of 2). 
  
  This shows that every balloon $B$ in the support of $y^*$ intersects
  $Y$ with at least as large a total coefficient as it intersects $Z$.
  Since $y^*$ is a packing that saturates $Z$ and does not over-saturate
  any vertex, we have $c(Y) \geq \sum_{\cB_1} y^*_B + \sum_{\cB_2} 2y^*_B = c(Z)$
  as promised. This finishes the proof.
\maybeqed
\end{proof}

\section{The FPT algorithms}

We finally wrap up by giving our main results.

\subsection{\localproblem}\label{sec:local}

We use the results of the previous section to finalise
Theorem~\ref{thm:local}. We proceed by lemmas. Throughout, we assume
access to a membership oracle for the biased graph
so that we can optimise the LP.

\begin{lemma}
  \localproblem admits a 2-approximation, even for weighted graphs.
\end{lemma}
\begin{proof}
  Let $x^* \in [0,1]^V$ be an LP-optimum computed via Theorem~\ref{thm:optimise}.
  Compute the sets $V_R$, $V_1$ and $V_{1/2}$ from $x^*$ as in Section~\ref{sec:halfint},
  forming a half-integral optimum $x$. It is now clear that the set
  $X=V_1 \cup V_{1/2}$ is an integral solution and a 2-approximation to the problem.
\maybeqed
\end{proof}

\begin{lemma} \label{lemma:thm:gap}
  \localproblem for unweighted graphs can be solved in $O^*(4^{k-\lambda})$ time,
  where $\lambda$ is the optimum of the local LP.
\end{lemma}
\begin{proof}
  Assume $k<n$, as otherwise we may simply accept the instance. 
  We will execute a branching process, repeatedly selecting half-integral vertices
  and recursively ``forcing'' their variables to take values $x_v=0$ or $x_v=1$. 

  More precisely, we will use the terms \emph{fix $v=0$} and \emph{fix $v=1$} for the 
  following procedures.
  To fix $v=0$, we simply set $c_v=2n$; since the LP is half-integral in the presence 
  of vertex weights, this implies that either $x_v=0$ in a half-integral optimum or the LP-optimum costs
  at least $n>k$.  To fix $v=1$, we similarly set $c_v=0$. 
  Let us make a quick observation about this procedure.

  \begin{claim}
    Consider a situation where we have fixed $v=0$ for a set of vertices $A_0$,
    and fixed $v=1$ for a set of vertices $A_1$, where $A_0, A_1 \subseteq V$.
    Let $\lambda$ be the resulting optimal value of the LP-relaxation.
    Then the following hold.
    \begin{enumerate}
    \item If there is a set $S \subseteq V \setminus (A_0 \cup A_1)$
      such that $A_1 \cup S$ is a solution to \localproblem,
      then $|S| \geq \lambda$
    \item If there is an integral solution to the LP of cost $\lambda < n$,
      setting $x_v=1$ for some set of vertices $S \subseteq V$, then
      $S \cap A_0 = \emptyset$, and $S \cup A_1$ is a solution to \localproblem
      of cardinality $|A_1|+\lambda$. 
    \end{enumerate}
  \end{claim}
  \begin{proof}[Proof of claim]
    The only difference incurred by fixing vertices is a change of vertex weights $c_v$.  
    Therefore, in the first case, if there is a set $S$ such that $S \cup A_1$ is a 
    solution, then setting $x_v=1$ for $v \in S \cup A_1$ is a feasible solution to the LP 
    of cost $|S|$, hence $|S| \geq \lambda$. In the second case, as observed above,
    any half-integral solution of cost $\lambda < n$ must set $x_v=0$ for $v \in A_0$.
    Therefore, the statement follows since the set of solutions is upwards closed. 
  \end{proof}

  We now describe the branching process. Consider a generic situation, 
  where some sets $A_0$ resp.~$A_1$ of vertices have been fixed to $v=0$
  resp.~$v=1$, and where we have remaining budget $k$ to spend on non-fixed vertices.
  Compute a half-integral optimum $V_1+\tfrac 1 2 V_{1/2}$ for the local LP as above,
  and reject the instance if $\lambda > k$.
  Adjust this optimum such that the corresponding set $V_R$ is 
  as large as possible, by repeatedly fixing $v=0$
  for non-fixed vertices $v \in V_{1/2}$, and keeping $v$ fixed if
  the cost of the LP does not increase. 
  Assume first that this leaves $V_{1/2}$ with only fixed vertices
  Then we must have $V_{1/2} \subseteq A_1$, since otherwise we would have $\lambda > k$, 
  and the solution where $x_v=1$ for every $v \in V_1 \cup V_{1/2} \cup A_1$ is an integral LP-optimum
  since fixed vertices have cost $c_v=0$.  In particular, we can compute 
  the set $V_R$ of reachable vertices around $v_0$ as previously, and $X=N(V_R)$
  will be a possibly smaller solution of at most $k+|A_1|$ vertices.

  Otherwise, fix $v=0$ for every vertex $v \in V_R$, 
  fix $v=1$ for every vertex $v \in V_1$, 
  and reduce $k$ by $|V_1 \setminus A_1|$.
  As observed previously, since $V_1 + \tfrac 1 2 V_{1/2}$
  is an LP-optimum of a persistent LP, these vertices can all be fixed 
  without losing any integral optimum. 
  Select an unfixed vertex $v \in V_{1/2}$, which exists by assumption, 
  and branch recursively on fixing $v=0$ and fixing $v=1$, in the latter branch 
  reducing our budget $k$ by one, and solve the problem recursively. 
  Then in the branch $v=0$, $\lambda$ increases by at least $1/2$,
  since the LP-relaxation has half-integral costs and fixing $v=0$ increases the cost;
  and in the branch $v=1$, $\lambda$ decreases by $1/2$ due to the change of vertex cost,
  whereas $k$ decreases by at least $1$. Hence in both branches the value of
  $k-\lambda$ decreases by at least $1/2$, which means that after 
  a branching depth of at most $2(k-\lambda)$, each branch has either
  terminated or produced a solution. 
  Thus the total running time is $O^*(2^{2(k-\lambda)})$ as promised.
\maybeqed
\end{proof}

\begin{lemma}
  \localproblem for unweighted graphs can be solved in $O^*(2^k)$ time.
\end{lemma}
\begin{proof}
  If $\lambda \leq k/2$, then we may produce a solution of size at most $k$
  by rounding up the LP-optimum. Otherwise $\lambda > k/2$ 
  and the result follows from Lemma~\ref{lemma:thm:gap} since
  $4^{k-\lambda}=2^{2(k-\lambda)} < 2^{2(k-k/2)}=2^k$. 
\maybeqed
\end{proof}

This concludes the proof of Theorem~\ref{thm:local}.

\subsection{\problem}

We now finally show the full solution to \problem.

\begin{proof}[Proof of Theorem~\ref{thm:global}]
  Select an arbitrary vertex $v_0 \in V$ and branch over two options: either delete $v_0$,
  or decide that $v_0 \notin X$ and proceed to solve the local problem rooted in $v_0$.  
  In the latter case, compute a half-integral optimum $V_1 + \tfrac 1 2 V_{1/2}$ of the local
  LP for which the set $V_R$ of reachable vertices is maximal, as in
  Lemma~\ref{lemma:thm:gap}. 
  We claim that under the assumption that there is an optimum $X \subseteq V$
  to the global problem with $v_0 \notin X$,
  there is such an optimum with 
  $V_1 \subseteq X$ and $V_R \subseteq V(H)$ for some connected component $H$ 
  of $G - X$. 
  
  For this, let $Y$ be an optimum with $v_0 \notin Y$, and let $H$ be the connected 
  component of $G - Y$  for which $v_0 \in V(H)$. Applying Lemma~\ref{lemma:persistence} 
  to $V(H)$ gives us the sets~$S^+=N[V(H) \cup V_R]$
  and $S' \supseteq V_R$. Let $Y' = (Y \setminus S^+) \cup (S^+ \setminus S')$. 
  Since $N[V(H)] \subseteq S^+$  we have $N(V(H)) \subseteq Y \cap S^+$, and 
  by Lemma~\ref{lemma:persistence} we have $|N(V(H))| \geq |S^+ \setminus S'|$,
  hence $|Y'| \leq |Y|$. We also have $V_1 \subseteq Y'$, and $G - Y'$ contains a connected
  component $H'$ with $V_R \subseteq V(H')$. We claim that $Y'$ is a solution. Assume to 
  the contrary that there is some unbalanced cycle $C$ with $V(C) \cap Y' = \emptyset$. 
  Then $C$ intersects $Y$ in $S^+ \setminus Y'=S'$. But since $N(S') \subseteq Y'$ 
  this contradicts that $G[S']$ is balanced. Hence $Y'$ is also an optimal
  solution, and the claim is shown. 

  Hence, we may fix $v=0$ for every $v \in V_R$, and $v=1$ for every $v \in V_1$,
  and proceed as in Lemma~\ref{lemma:thm:gap} until the vertices fixed to $0$
  contain a connected component containing $v_0$, surrounded entirely by vertices
  fixed to~$1$. In such a case, we simply proceed as above with a new starting vertex 
  $v_0$ in a non-balanced connected component of $G$,
  until we either exceed our budget $k$ or discover an integral solution $X$,
  and we are done. As in Lemma~\ref{lemma:thm:gap}, while branching on a local LP
  the gap between lower bound and remaining budget decreases in both branches, 
  whereas branching on a new vertex $v_0$ will certainly increase the solution cost,
  since the previous solution at this point does not account for any vertices
  in the connected component of $v_0$. 
  Hence we have a tree with a branching factor of 2 and a height of at most $2k$,
  implying a total size and running time of $O^*(4^k)$, and we are done.
\maybeqed
\end{proof}

\subsection{Lower bounds}

We show two lower bounds for \problem, one unconditional in the
black-box oracle model, and one conditional on SETH (the strong
exponential-time hypothesis).  More concretely, let us say that an
algorithm solves \problem in the \emph{black-box oracle model}
if it takes as input $(\Psi=(G=(V,E), \cB), k)$ where $\Psi$ is a
biased graph and $\cB$ is provided purely in the form of a membership
oracle. The result is simply the more carefully worked-out version of
the lower bound informally announced in the introduction. 

\begin{theorem} \label{thm:oracle-lb}
  Every algorithm that solves \problem in the black-box oracle model
  needs to make $\Omega(2^k)$ membership queries in the worst case,
  both in the edge- and vertex-deletion versions. 
\end{theorem}
\begin{proof}
  We show the result for the edge-deletion version.  It is standard to
  reduce this to the vertex-deletion version by a suitable transformation.
  
  Describe a graph $G_k$ as follows. Start with the simple cycle $C_k$,
  with vertex set $V=\{v_1,\ldots,v_k\}$ and edges $v_iv_{i+1}$, where
  $v_{k+1}=v_1$ (i.e., vertex numbers are read cyclically along $C_k$).
  then replace every edge $uv$ in $C_k$ by a $C_4$, on vertices $u-x^1_{uv}-v-x^2_{uv}-u$, where $x^1_{uv}$, $x^2_{uv}$ are new vertices. Pick an arbitrary vector $b \in \{0,1\}^k$.
  We next define a biased graph $\Psi_{b,k}=(G_k, \cB_b)$ where $\cB_b$
  consists of the single cycle $C_b$ on vertex set
  $V(C_b)=\{v_1, \ldots, v_k\} \cup \{x_{v_iv_{i+1}}^{b_i}: i \in [k]\}$,
  where again $v_{k+1}=v_1$. That is, precisely one out of all $2^k$
  ways of traversing the original cycle $C_k$ is considered balanced.
  It is easy to see that $\Psi_k$ is a biased graph, since every
  unbalanced cycle $C$ and chord path $P$ for $C$ contains two
  distinct possible reconfigurations of $C$, one of which is distinct
  from $C_b$.  We claim that $(\Psi_{b,k},k)$ is positive for every
  vector $b$, but has an essentially unique solution. Indeed, since
  $G_k$ contains $k$ pairwise edge-disjoint unbalanced $C_4$'s, every
  deletion set of size $k$ must contain precisely one edge for each
  $C_4$, and this leaves precisely one cycle $C_{b'}$ spanning $C_k$,
  where $C_{b'}$ is defined as $C_b$ according to some vector $b' \in \{0,1\}^k$. 
  Since the choice of $b$ was arbitrary, the only way to detect in a
  black-box fashion which choice of deletions is feasible is to probe
  each such cycle $C_{b'}$ until the choice $b'=b$ is found.
\end{proof}

Via known lower bounds for \textsc{Unique $k$-SAT}, we can also show
a similar lower bound for explicitly represented instances assuming
SETH (the Strong Exponential-Time Hypothesis)~\cite{CalabroIKP08JCSS,ImpagliazzoP01,CyganDLMNOPSW16TALG}.
Recall that SETH is the hypothesis that for every $\varepsilon > 0$
there is a $k \in \N$ such that $k$-SAT cannot be solved in time
$O^*(2^{(1-\varepsilon)n})$.

\begin{corollary}
  Unless SETH is false, there is no algorithm that solves \problem
  in time $O^*((2-\varepsilon)^k)$ for any $\varepsilon > 0$, even if
  the class of balanced cycles of the input graph is provided 
  through an explicit circuit. 
\end{corollary}
\begin{proof}
  \textsc{Unique $k$-SAT} is the promise problem where the input is a
  $k$-CNF formula $F$ with the promise that $F$ has at most one
  satisfying assignment, and the question is whether $F$ is
  satisfiable. Calabro et al.~\cite{CalabroIKP08JCSS} show that under
  SETH, $k$-SAT and \textsc{Unique $k$-SAT} have the same asymptotic
  complexity.  In other words, for every $\varepsilon > 0$ there is a
  $k \in \N$ such that if \textsc{Unique $k$-SAT} can be solved in
  time $O^*(2^{(1-\varepsilon)n})$, then SETH is false.
  
  Assume that \problem can be solved in time $O^*((2-\varepsilon)^k)$
  for some $\varepsilon > 0$, and let $q \in \N$ be such that
  \textsc{Unique $q$-SAT} cannot be solved in $O^*(2^{(1-\varepsilon)n})$
  time under SETH. Let $F$ be a an instance of \textsc{Unique $q$-SAT},
  and let $k$ be the number of variables in $F$. Fix an arbitrary
  ordering on the variables of $F$. 
  Construct the graph $G_k$ as in Theorem~\ref{thm:oracle-lb},
  and let the class of balanced cycles $\cB$ be provided as a
  membership oracle, where a cycle $C_b$, $b \in \{0,1\}^k$ is
  balanced if and only the assignment $b$ satisfies $F$ (interpreted
  according to the ordering of variables of $F$). Then clearly,
  $\cB$ can be provided as a polynomial-time membership oracle,
  and due to the promise property of \textsc{Unique $q$-SAT}
  the resulting class $\cB$ is a balanced class of cycles.
  Furthermore, the instance $(\Psi=(G_k,\cB),k)$ has a solution if and
  only if $F$ has a satisfying assignment.  
\end{proof}

Finally, regarding the special case of \textsc{Group Feedback Vertex Set} (GFVS),
we note that Theorem~\ref{thm:oracle-lb} can easily be extended to give the
same lower bound for GFVS when the group is represented in a black-box
manner using the group $Z_2^n$ (see Cygan et al.~\cite{CyganPP16GFVS}
for a discussion on black-box representations of GFVS instances). 
However, we are not aware of any lower bound for the running time of
\textsc{Group Feedback Vertex Set} in any explicit group representation.

\section{Polynomial time approximation}
In this section, we give a polynomial time $O(\log k)$-approximation algorithm for the weighted version of \problem,
where $k$ denotes the cardinality of the optimal solution. 
The algorithm uses the standard {\em region growing} argument used for \textsc{Multicut} and \textsc{Sparsest Cut}~\cite{LeightonR99, GargVY96}, with a simple observation that leads to bound the approximation ratio by $O(\log k)$ instead of $O(\log n)$. 
Given a weighted biased graph $\Psi = (G = (V, E), \cB)$ with weights $c : V \to \Q^+$ and the corresponding class of unbalanced simple cycles $\cC$, we consider the following {\em global} LP relaxation that has $\{ x_v \}_{v \in V}$ as variables.

\begin{align}
\mbox{minimize }& \sum_{v \in V} c_v x_v \\
\mbox{subject to }& \sum_{v \in C} x_v \geq 1 && \forall C \in \cC \\
& \sum_{v \in V} x_v \leq k &&  \\ 
& x_v \geq 0 && \forall v \in V
\end{align}

\begin{lemma}
Given access to a polynomial-time membership oracle for $\cB$, the global LP relaxation admits a polynomial-time separation oracle. 
\end{lemma}
\begin{proof}
  By Theorem~\ref{thm:optimise} there is a separation oracle for the
  balloon constraints of the local LP, for any root vertex $v_0 \in
  V(G)$. We reuse these as constraints in the global LP as follows. 
  Let $uv \in E(G)$ be an arbitrary edge, and let $v_0$ be a new
  vertex subdividing $uv$ in $G$.  Set $x_{v_0}=0$.  We now run the
  separation oracle for the local LP rooted in $v_0$.  Repeating this
  for every edge of $G$ will yield our separation oracle for the
  global LP.  The oracle is clearly polynomial time; we show
  correctness. 

  By Lemma~\ref{lemma:decompositions}, for every $v_0$-balloon $B$
  involved in a constraint $w(B) \geq 1$, the balloon $B$ contains an
  unbalanced cycle $C$ and $w(B) \geq w(C):=\sum_{v \in C} x_v$.
  Therefore every constraint $w(B) \geq 1$ of the local LP rooted in
  $v_0$ is a valid constraint for the global LP.  In the other
  direction, let $C$ be an unbalanced cycle such hat $w(C)<1$ and pick
  a root $v_0$ on an edge $uv$ of $C$. Then the balloon $B=(v_0,C)$ is
  a valid $v_0$-balloon, where $v_0$ represents the path at $v_0$ with
  no edges. Thus $w(B)=w(C)<1$, and the local LP rooted at $v_0$
  contains a violated constraint $w(B) \geq 1$. One such constraint
  will be found by the separation oracle.
\end{proof}

We define quantities similar to previous sections in a slightly different manner. 
For a path $P$, we let $\ell_x(P) = \sum_{v \in P} x_v$. For $u, v \in V$, we let $z_x(u, v) = \min_P \ell_x(P)$ ranging over all $u$-$v$-path $P$. 
Note that the current definition of $\ell_x(P)$ differs from the previous definition by $\frac{1}{2}(x_u + x_v)$, where $u, v$ are the endpoints of $P$. 
We omit subscripts $x$ if they are clear from the context. 
The {\em region} around $v$ of radius $r$ is denoted by $R_{v, r}$ and consists of the following two components.
\begin{itemize}
\item Interior $I_{v, r} = \{ u \in V : z(u, v) < r \}$.
\item Boundary $\partial R_{v, r} = \{ u \in V : z(u, v) - x_u < r \leq z(u, v)  \}$. 
\end{itemize}

The LP constraints ensure that an interior of small radius does not have an unbalanced cycle. 

\begin{lemma}
For any $v \in V$ and $r \leq 1/4$, $I_{v, r}$ does not have any unbalanced cycle. 
\end{lemma}
\begin{proof}
  Consider the alternative assignment $x'$ where $x_u'=x_u$  for every
  $u \in I_{v,r}$ and $x_u'=1$ otherwise. Assume that $I_{v,r}$
  contains an unbalanced cycle, and let $C$ be such a cycle of minimum
  weight.  
  We claim that for any pair of vertices $u, w \in C$, there is a
  shortest $u$-$w$-path (under the distance $z_{x'}(u,w)$) that
  follows $C$.  Assume to the contrary and let $u, w \in V$ be
  vertices such that both $u$-$w$-paths along $C$ are longer than
  $z_{x'}(u,w)$, with $u, w$ chosen to minimise $z_{x'}(u,w)$. 
  Let $P$ be a shortest $u$-$w$-path. By the choice of $u, w$,
  only the endpoints of $P$ lie in $C$, i.e., $P$ is a chord path for
  $C$.  But then both possible results of reconfiguring $C$ using $P$
  yield a cycle $C'$ shorter than $C$. This contradicts the choice of
  $C$.

  Now pick a vertex $u \in C$, and let $w_1, w_2 \in C$ be neighbours in $C$
  such that $C$ decomposes into a shortest $u$-$w_1$-path and a shortest
  $u$-$w_2$-path.  Observe that $w_1, w_2$ must exist.
  Then the weight of $C$ is $w(C):=\sum_{w \in C} x_w = z_{x'}(u,w_1)+z_{x'}(u,w_2)-x_u$.
  But since $I_{v,r}$ has radius less than $1/4$, all distances within $I_{v,r}$
  are less than $1/2$.  Thus $w(C) \leq z_{x'}(u,w_1)+z_{x'}(u,w_2) < 1$.
  Since the weight of $C$ is equal under $x$ and $x'$, the cycle $C$
  represents a constraint violated in the LP.
\end{proof}

For $U \subseteq V$, let $x(U) = \sum_{u \in U} x_u$, 
$c(U) = \sum_{u \in U} c_u$, 
$LP(U) = \sum_{u \in U} c_u x_u$.
Intuitively, a region $R_{v, r}$ entirely contains a vertex $u \in I_{v, r}$ in its interior, and for $w \in \partial R_{v, r}$, only a $\frac{r - (z(w, v) - x_w)}{x_w}$ fraction of it is contained in $R_{v, r}$. 
To be consistent with this intuition, let
\begin{align*}
x(R_{v, r}) &= x(I_{v, r}) + \sum_{w \in \partial R_{v, r}} (r - (z(w, v) - x_w)) \\
LP(R_{v, r}) &= LP(I_{v, r}) + \sum_{w \in \partial R_{v, r}} c_w (r - (z(w, v) - x_w)).
\end{align*}

After obtaining the LP solution $x$, the rounding algorithm proceeds as follows. 
The graph $G = (V, E)$ is modified throughout the algorithm. 
Let $\opt = LP(V)$ be the initial LP value. 
Other definitions including $R_{v, r}$ are with respect to the current graph. 
Note that for any $U \subseteq V$, the LP solution $x$ restricted to $U$ is a feasible solution for $G[U]$. 

\begin{enumerate}
\item Choose an arbitrary vertex $v \in V$. 
\item Find the smallest $r \geq 1/8$ such that $c(\partial R_{v, r}) \leq (16 \ln k) \cdot  (LP(R_{v, r}) + \opt / k)$. 
\item Delete $\partial R_{v, r}$, which will separate $I_{v, r}$ from the rest of the graph. Let $V \leftarrow V \setminus (I_{v,r} \cup \partial R _{v, r})$. 
\item Repeat from Step 1 as long as $V \neq \emptyset$. 
\end{enumerate}

The standard analysis of the region growing process ensures the following fact. 
\begin{lemma}
In Step 2, above, the smallest $r$ is always at most $1/4$. 
\end{lemma}
\begin{proof}
Let $y_r := LP(R_{v, r}) + \opt / k$, and note that $c(\partial R_{v, r}) = \frac{d y_r}{dr}$. If the smallest $r$ is greater than $1/4$, 
$\frac{d y_r}{dr} = c(\partial R_{v, r}) > (16 \ln k) y_r$ for $r \in [1/8, 1/4]$, which implies 
\[
y_{1/4} > e^{(16 \ln k) \cdot (1/8 - 1/4)} \cdot y_{1/8} = k^2 \cdot y_{1/8}. 
\]
Since $y_{1/4} \leq \opt(1 + 1/k)$ and $y_{1/8} \geq \opt/k$, it leads to contradiction when $k \geq 2$. 
\end{proof}

Therefore, each interior $I_{v, r}$ separated from the graph has no unbalanced cycle. 
Furthermore, notice that $x(R_{v, r}) \geq r \geq 1/8$ and all regions are disjoint, so
the constraint $\sum_v x_v \leq k$ implies that the above algorithm runs in at most $8k$ iterations.
In each iteration, the weight of deleted vertices is at most $(16 \ln k) \cdot (LP(R_{v, r}) + \opt / k)$.
Since the sum of the first terms is at most $\opt$ and there are at most $8k$ iterations, 
the total weight of deleted vertices is at most 
\[
(16 \ln k) \cdot \bigg( \opt + 8k \cdot \opt / k \bigg) \leq O(\ln k) \cdot \opt,
\]
giving an $O(\log k)$-approximation algorithm for \problem. 

\section{Conclusions}

We have shown that the combinatorial notion of \emph{biased graphs}, 
especially the notion of \emph{co-linear cycle classes}, allows
us to formulate an LP-branching FPT algorithm for a surprisingly
broad class of problems, including the full generality of the
\problem parameterized by $k$, and \localproblem parameterized 
by relaxation gap. 
Compared to previous results~\cite{IwataWY16}, these algorithms
are somewhat more general, and significantly more grounded in
combinatorial notions.
We also showed that \problem admits an $O(\log k)$-approximation,
where $k$ is the solution size. 

Open problems include completely combinatorial FPT algorithms,
and settling the associated kernelization questions. 

\bibliographystyle{abbrv}
\bibliography{all}
\end{document}